\documentclass[conference]{IEEEtran}
\usepackage{amssymb}
\usepackage{amsmath}
\usepackage{amsthm}
\usepackage{graphicx}
\usepackage{subfigure}
\usepackage{hyperref}

\def\abstract{
\typeout{Abstract}
 {\bf Abstract} 
} 

\newcommand{\be}{\begin{equation}}
\newcommand{\ee}{\end{equation}}

\newcommand{\hmu}{\hat{\mu}}
\newcommand{\hsigma}{\hat{\sigma}}

\newtheorem{thm}{Theorem}

\begin{document}

\title{Early Stopping Based on Repeated Significance}

\author{Eric Bax, Arundhyoti Sarkar, and Alex Shtoff\\Yahoo!}

%\editor{??}

\maketitle

\begin{abstract}
For a bucket test with a single criterion for success and a fixed number of samples or testing period, requiring a $p$-value less than a specified value of $\alpha$ for the success criterion produces statistical confidence at level $1 - \alpha$. For multiple criteria, a Bonferroni correction that partitions $\alpha$ among the criteria produces statistical confidence, at the cost of requiring lower $p$-values for each criterion. The same concept can be applied to decisions about early stopping, but that can lead to strict requirements for $p$-values. We show how to address that challenge by requiring criteria to be successful at multiple decision points.
\end{abstract}

%\begin{IEEEkeywords}
%error bound, validation, ensemble, Gibbs classifier, average error bound
%\end{IEEEkeywords}

%\begin{keywords} 
%validation, ensemble, classifier, error bound, average
%\end{keywords}

\section{Introduction}
In an \textit{AB test}, we compare two treatments. In online user studies, we often refer to AB tests as \textit{bucket tests}, with each user assigned to a bucket -- a set of users given the same treatment. When comparing a new treatment to an established one, we refer to the set of users receiving the new treatment as the \textit{test bucket} and the set of users receiving the established treatment as the \textit{control bucket}. Buckets may also be called \textit{arms} of a study -- a common usage in medical studies.

In the simplest case, we decide before running an AB test how long to run it (in terms of time or number of observations), a single criterion for treatment success, such as the statement ``the new treatment increases revenue per page view compared to the established treatment," and a desired level of confidence, for example 99\%. Then, when the test completes, we get a $p$-value for the statement from the AB-testing system, and if it is 1\% or less, then we roll out the new treatment to all users. In this simple case, we can subtract our desired level of confidence (99\%) from 100\% to get the $p$-value required to achieve that confidence: 100\% - 99\% = 1\%, which may be expressed as $p \leq 0.01$. 

This simple case does not fully capture how analysts work with buckets. In some cases, we want to be confident (statistically speaking) that multiple statements are true, for example we may want 95\% confidence that the new treatment increases revenue while also not increasing the rate of short dwell-time clicks. We may also want to observe the bucket results periodically, and, if we can do so with confidence, stop the bucket early, either confident that we should implement the new treatment or that we should avoid it. In product optimization, early stopping speeds time to market for individual features and also allows more potential features to be tested. In medical studies, early stopping allows people in the control arm of a study to receive useful treatments faster, or it stops administering failed treatments to people in the experimental arm more quickly than running the test to completion.

However, early stopping introduces the potential for a temporal form of data dredging \cite{smith02,wasserstein02,young11}: selecting a stopping time based on indicators of statistical significance affects whether those indicators are accurate. A set of methods to address that challenge are called \textit{always-valid bounds} or \textit{continuous monitoring methods} \cite{deng16,johari21,maharaj23,grunwald24,waudbysmith24}. They give stopping conditions that offer confidence for any number of samples or at any time during an AB test. The methods in this paper give stopping conditions that offer confidence for a limited number of potential stopping points, called decision points, so they are \textit{group sequential methods} \cite{wald47,peto76,pocock77,wang87,jennison00,lewis23}, which are also called \textit{interim analysis methods} \cite{huang17,grayling18,ciolino23}. The methods in this paper use the group sequential method strategy of ``spending" a type I error budget \cite{lan83,lewis23}. To ensure generality and ease of use, the methods emphasize simplicity -- they do not rely on assumptions about the joint distributions of $p$-values among successive decision points or across different criteria, so they avoid the need to confirm those assumptions for each AB test.

In this paper, we offer some analysis to help the practitioner decide when to use the methods described here. For 20 or fewer decision points during an AB test (eg a three-week test with daily decisions after the first week), we show that the simplest methods described here tend to outperform always-valid bounds. For more decision points, we show that requiring a small time period of continued success before stopping tends to make the methods in this paper more effective than always-valid bounds. 

This note also offers some strategies to develop test plans for AB testing with multiple criteria and the potential for early stopping. The strategies are straightforward, allowing the practitioner to understand why they work, to mediate tradeoffs between tactics, and to create test plans with confidence. 

Section \ref{sec_bonferroni} reviews tools to bound probabilities of conjunctions of multiple events and how to apply those tools to testing with multiple success criteria. Section \ref{sec_early} reviews how to apply those ideas to sequential testing -- testing with the possibility of early stopping while controlling type I error (false positives). Section \ref{sec_rep} examines how requiring multiple statistically significant results before stopping a test early can ease the $p$-value requirements for those results. Section \ref{sec_cm} explores how to apply repetition requirements effectively as the number of decision points increases, even up to a potential decision point after each new observation: the realm of continuous monitoring. Section \ref{sec_type2} analyzes how to control type II error (false negatives) via $\alpha$-spending methods combined with requiring repetition for early stopping. Section \ref{sec_disc} concludes this paper with a discussion of directions for future work.

\section{Multiple Criteria, Bonferroni, and Boole} \label{sec_bonferroni}
Suppose you and your five closest friends rent some bicycles for the afternoon. If there is a 5\% chance of each bicycle breaking down during your ride, how (statistically) confident are you that the trip will be completed without incident? 

The answer depends on the relationships between the breakdown probabilities for the different bicycles. Suppose breakdowns are completely correlated -- for example, we ride together, and that 5\% is the probability that we will choose to bike through an area with so much broken glass that we are all certain to get flat tires. Then we are 95\% confident that we will all avoid breakdowns. That is the best case. 

In the worst case, the failure events are disjoint -- if one bicycle breaks, the others do not. For example, suppose the bicycle shop has 20 bicycles, 19 never break down, and one breaks down on every trip. By taking out six bicycles at random, the probability that we select the bad bicycle is 6 out of 20, which is 30\%. In this worst case, each bicycle has a 5\% chance of being the bad one, so the probability that one of our six is the bad one is six times 5\%, which is 30\%. That leaves us only 70\% confident of an incident-free outing.  (See Figure \ref{uniform}.)

\begin{figure} 
\includegraphics[width=3.5in]{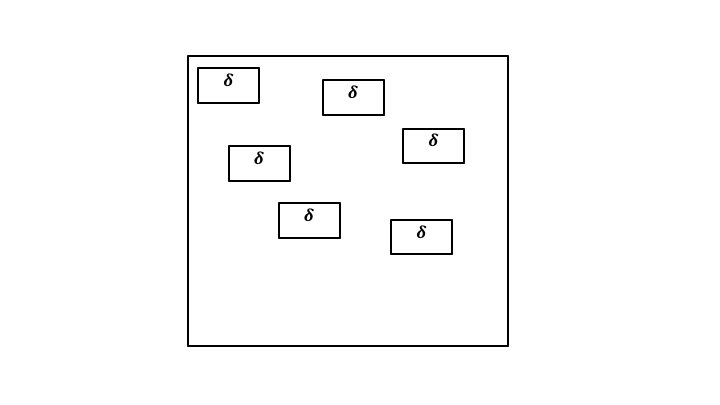} 
\caption{\textbf{Worst Case}. If six statements each have probability at most $\delta$ of being incorrect, then the probability that the combined statement: ``statement 1 and statement 2 and $\ldots$ and statement 6" is incorrect may be as high as $6 \delta$, because only one statement has to be incorrect for the ``and" of all 6 to be incorrect. The worst case is that failures are disjoint -- like how spreading six carpets so that they do not overlap covers as much area as possible. (Figure from \cite{bax16}.)} \label{uniform}
\end{figure}

Without deeper insight into the relationships between failures, we need to assume the worst case: that failure probabilities for individual criteria sum to give the failure probability for the conjunction of the criteria. So, to maintain a confidence level of $1 - \alpha$ over multiple criteria, we have to require that their $p$-values sum to $\alpha$ or less. If $m$ is the number of statements, then we can require each $p$-value to be $\frac{\alpha}{m}$ or less:
\be
p \leq \frac{\alpha}{m}.
\ee 
If that occurs, then we have our desired confidence in the conjunction of the criteria. 

As an example, suppose we want a bucket test to yield 95\% confidence that we meet both a revenue criterion and a user-experience criterion before implementing a new feature for a website. Then the allowed failure probability is $\alpha = 0.05$, since 100\% - 95\% = 5\%. So if we require the $p$-values for both statements to be 2.5\% (= 0.025) or less, then we can proceed to roll out the treatment with 95\% confidence that both the revenue and user-experience criteria will be met. 

Dividing $\alpha$ by the number of criteria $m$ to get the required $p$-value for each criterion is called a Bonferroni correction \cite{bonferroni36}. As we have shown, it is based on the idea that the sum of probabilities of events is a bound on the probability of the union of events. (If the events are disjoint, then that bound is the probability of the union.) That sum bound is known as Boole's inequality \cite{boole47,casella01}.

We can also partition $\alpha$ unequally over the criteria and still achieve $1 - \alpha$ confidence. In general, for any $\alpha_1 + \ldots + \alpha_m = 1$ with all $\alpha_i$ nonnegative, if we require 
\be
p \leq \alpha_i
\ee
for each criterion $i$, then all criteria hold, with confidence $1 - \alpha$. Equivalently, the probability of type I error, also called a false-positive result, is at most $\alpha$. To see why, consider that the probability of a type I error for criterion $i$ is at most $\alpha_i$, so we can apply the sum bound to show that the probability of a type I error for any criterion is at most the sum $\alpha_1 + \ldots + \alpha_m$. In a sense, $\alpha$ is a type I error ``budget" that we ``spend" \cite{lan83,lewis23} over the criteria to get required $p$-values.

\section{Early Stopping} \label{sec_early}
Suppose we plan to run an AB test for up to two weeks, and we will check the bucket criteria each day, with the check based on data collected since the start of the test. Based on a daily check, we may decide to stop the AB test early. Then the test has a decision point each day. (To clarify, measurements for the bucket criteria may be observed more often than daily and still we have only one decision point per day if the other observations do not contribute to stopping decisions.) The data boundaries for decisions must be determined a priori, for example measurements over the first 24 hours of data, then the first 48 hours of data, etc. Similarly, for a test that is to run until a thousand observations are gathered, the decision points can be after gathering each hundred observations. (Decision points are sometimes called \textit{interim analyses} in clinical trial literature \cite{huang17,grayling18,ciolino23,lewis23}.) In a later section, we discuss how to handle AB testing without determining the time or amount of data for the entire test a priori, and hence potentially having an unknown number of decision points. 

Here are some theorems that enable early stopping with confidence:

\begin{thm} \label{thm_dm}
Let $d$ be the number of decision points and $m$ be the number of bucket criteria. Then requiring $p$-values
\be
p \leq \frac{\alpha}{d m}
\ee
at least once for each criteria $i$ gives confidence $1 - \alpha$ that all criteria hold.
\end{thm}

\begin{proof}
Apply the Bonferroni inequality. The probability that the AB testing system incorrectly reports that a criterion holds at a decision point with a $p$-value of $\frac{\alpha}{d m}$ or less is at most $\frac{\alpha}{d m}$. So the probability of any incorrect assertion with a $p$-value of $\frac{\alpha}{d m}$ or less, about any criterion at any decision point, is at most the sum of those probabilities over the number of criteria times the number of decision points: 
\be
d m \frac{\alpha}{d m} = \alpha.
\ee
\end{proof}

Using a more general approach to budgeting $\alpha$ and allowing different numbers of decision points for different criteria gives a more general theorem:

\begin{thm} \label{thm_ait}
Let $m$ be the number of criteria and let $d_i$ be the number of decision points for criterion $i$. Let 
\be
\sum_{i = 1}^{m} \sum_{t=1}^{d_i} \alpha_{it} = \alpha
\ee
for nonnegative values $\alpha_{it}$. Then requiring $p$-values
\be
p \leq \alpha_{it}
\ee
at at least one decision point for each criterion gives confidence $1 - \alpha$ that all criteria hold.
\end{thm}

\begin{proof}
The probability that criterion $i$ has $p$-value $\alpha_{it}$ or less at decision point $t$ and the criterion does not hold out of sample is at most $\alpha_{it}$. Applying the sum bound (Boole's inequality), the probability that any criterion $i$ has $p$-value $\alpha_{it}$ or less at any decision point $t$ and the criterion does not hold is at most the sum of $\alpha_{i t}$ values: $\alpha$. So, if each criterion $i$ meets its $p$-value requirement at any of its decision points, then all criteria hold out of sample, with probability at least $1 - \alpha$.
\end{proof}

Theorems \ref{thm_dm} and \ref{thm_ait} allow early stopping because they only require each criterion to meet its $p$-value requirement once. After that, we have our required confidence without needing to continue the test.

\section{Requiring Repetition} \label{sec_rep}
With many decision points $d$ and several criteria $m$ for bucket success, the product $dm$ can be large. So requiring $p$-values below $\frac{\alpha}{dm}$ in Theorem \ref{thm_dm} may require more resources than we wish. Similarly, budgeting $\alpha$ over multiple decision points for each of multiple criteria in Theorem \ref{thm_ait} may leave us with small values $\alpha_{it}$, meaning stringent $p$-value requirements, and hence possibly requiring a larger or longer test than we wish. There is a way to readjust the odds in our favor, so to speak, and we can do so by supplying a mathematical foundation for something that many analysts already do: require conditions to be met over several decision points before making a final decision. 

Recall our example of bicycles that fail with 5\% probability each. Suppose you and three friends are taking some bicycles for a week of riding out in the desert. Together, the four of you take six bicycles, load them on a trailer, and drive it out to your campsite. You never bike so far that you cannot walk back if a bicycle fails, but you want to have enough bicycles that all four of you can go out on the next ride together even if some bicycles fail on previous rides. Since you have four people and six bicycles, having two bicycles fail is acceptable -- you still have four left, so everyone can still ride. 

How confident, statistically speaking, can you be that you can all keep riding, given that each bicycle has a 5\% probability of failing during the trip? Since you can allow two bicycle failures and still all ride, overall failure is now defined as having three or more bicycles fail. The worst-case relationship among bicycle failures is that any one failure is part of three failing on the same trip. In that case, the probability of three or more failing is 5\% times 2 = 10\%. (See Figure \ref{fig_nearly_uniform}.) So we have 90\% confidence that at least four bicycles will survive the trip. (The idea of allowing some failures in order to improve confidence is sometimes called nearly uniform validation \cite{bax_compression,bax16}.)

\begin{figure} 
\includegraphics[width=3.5in]{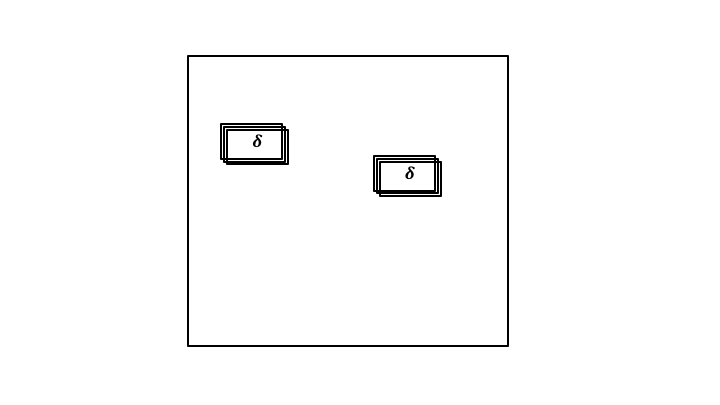} 
\caption{\textbf{Allowing Failures/Requiring Repetition}. If six statements each have probability at most $\delta$ of being incorrect, then the worst-case probability that three or more are incorrect is $2 \delta$. The worst case is that any failure is simultaneous with two others -- with six carpets, laying them three-thick only covers an area equal to two carpets. (Figure from \cite{bax16}.)} \label{fig_nearly_uniform}
\end{figure}

In general, the probability of $r$ or more events is bounded by the sum of event probabilities divided by $r$. In the worst case, each event occurrence is part of a simultaneous occurrence of exactly $r$ events. Imagine having a type of carpet you hate, because it is so deep and squishy that it is difficult to walk upon. If you have 6 of those carpets and each can cover 5\% of your floor, then together they might cover 30\% of your floor. But if you decide you don't mind the depth and squishiness unless you have to step on three of those carpets piled on top of each other, then the worst case is two piles of three that each cover 5\%, for a total of 10\% of your floor. 

Stated as a theorem:

\begin{thm} \label{thm_rdm}
Let $d$ be the number of decision points and $m$ be the number of bucket criteria. Then requiring $p$-values
\be
p \leq \frac{\alpha r}{d m}
\ee
at least $r$ times for each criterion gives confidence $1 - \alpha$ that all criteria hold.
\end{thm}

\begin{proof}
Use nearly uniform validation. The probability of the event that the AB testing system incorrectly reports that a criterion holds at a decision point with a $p$-value of $\frac{\alpha r}{d m}$ or less is at most $\frac{\alpha r}{d m}$. So the probability of at least $r$ such events is the number of events divided by $r$: 
\be
\frac{d m}{r} \frac{\alpha r}{d m} = \alpha.
\ee
If there are fewer than $r$ such events, then requiring that each criterion hold $r$ times means that the AB system is correct about each criterion at least once. 
\end{proof}

\begin{figure} 
\includegraphics[width=3.5in]{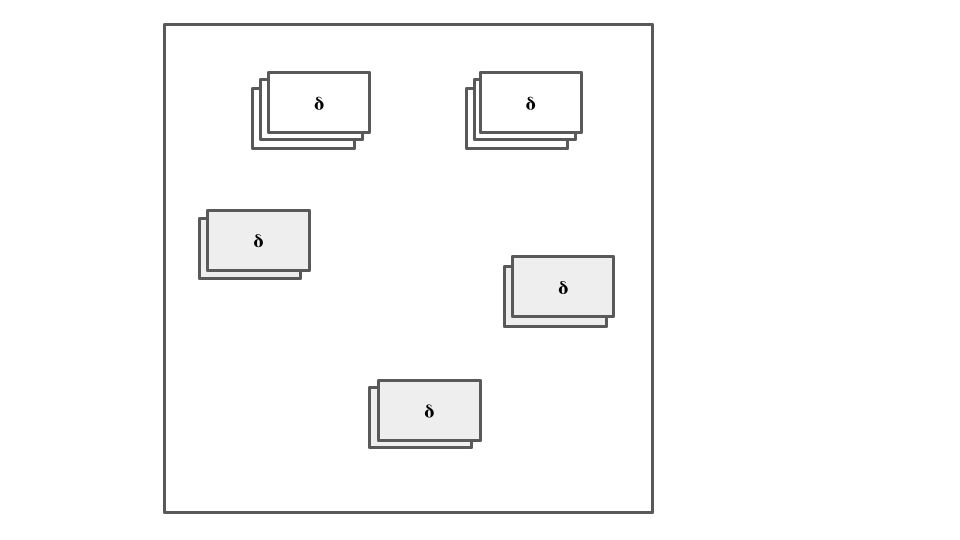} 
\caption{\textbf{Different Numbers of Repetitions}. Suppose we have two criteria for bucket success, and we require the first to hold three times and the second to hold twice to declare the test a success. Suppose there are $d = 6$ decision points and each statement that a criterion holds has probability $\delta$ of being incorrect. Then an incorrect conclusion from the bucket requires either incorrect statements that criterion one holds at three different decision points (white ``carpets" piled three-deep) or incorrect statements that criterion two holds at two different decision points (gray ``carpets" piled two-deep). So the probability of bucket success without the criteria actually both holding is at most $5 \delta$.} \label{fig_two_over_time}
\end{figure}

For a more general approach to $\alpha$-budgeting, allow different numbers of decision points for different criteria, and require different numbers of repetitions for different criteria. (Refer to Figure \ref{fig_two_over_time} for a specific example.)

\begin{thm} \label{thm_aitr}
Let $m$ be the number of criteria, let $d_i$ be the number of decision points for criterion $i$, and let $r_i$ be the number of repetitions required for criterion $i$.  Let 
\be
\sum_{i = 1}^{m} \sum_{t=1}^{d_i} \alpha_{it} = \alpha
\ee
for nonnegative values $\alpha_{it}$. Then requiring $p$-values
\be
p \leq \alpha_{it} r_i
\ee
at at least $r_i$ decision points for each criterion $i$ gives confidence $1 - \alpha$ that all criteria hold.
\end{thm}

\begin{proof}
The probability that criterion $i$ has $p$-value $\alpha_{it} r_i$ or less at decision point $t$ and the criterion does not hold out of sample is at most $\alpha_{it} r_i$. So, by nearly uniform validation, the probability that criterion $i$ has $p$-value $\alpha_{it} r_i$ or less at $r_i$ or more decision points, and the criterion does not hold out of sample, is at most
\be
\frac{\sum_{t = 1}^{d_i} \alpha_{it} r_i}{r_i} = \sum_{t = 1}^{d_i} \alpha_{it}.
\ee
Using the sum bound (Boole's inequality), the probability that that occurs for at least one of the $m$ criteria is at most the sum of those probabilities:
\be
\sum_{i = 1}^{m} \sum_{t=1}^{d_i} \alpha_{it} = \alpha.
\ee
So, if each criterion $i$ meets its $p$-value requirement at any $r_i$ of its decision points, then all criteria hold out of sample, with probability at least $1 - \alpha$.
\end{proof}

Theorems \ref{thm_rdm} and \ref{thm_aitr} allow early stopping: after each criterion meets its $p$-value requirement the required number of times, the test can be stopped with probability at most $\alpha$ of a type I error. Theorems \ref{thm_rdm} and \ref{thm_aitr} exchange less stringent $p$-value requirements than Theorems \ref{thm_dm} and \ref{thm_ait} for requiring repetition.

Theorem \ref{thm_aitr} can be useful if some criteria are likely to require more time or samples while others are likely to be met easily and continue to hold. The decisions about how many repetitions to require need to be made a priori, but in many cases information is available about how quickly $p$-values tend to decrease and settle for different metrics and their criteria. That information enables a priori optimization of the test plan. 

\section{Continuous Monitoring and General Result} \label{sec_cm}
Consider an AB test with a single criterion ($m = 1$) that may last as long as 20 days. Suppose we decide to divide $\alpha$ equally over $d$ decision points and require 5\% repetition: $r = 0.05 d$. Then, by Theorem \ref{thm_rdm}, the $p$-value requirements are 
\be
p \leq \frac{\alpha}{d} r = \frac{\alpha}{d} 0.05 d = 0.05 \alpha = \frac{\alpha}{20}.
\ee
Note that the required $p$-value is independent of the number of decision points. It is the same whether the decision points are daily ($d = 20$), hourly ($d = 480$), or with each new observation, which is the setting for continuous monitoring \cite{deng16,johari21,maharaj23,grunwald24,waudbysmith24}. In each case, the test yields the same level of confidence: $1 - \alpha$. In short, requiring significance for a fixed fraction of decision points allows us to increase the number of decision points without limit and maintain the same level of confidence. 

In general, for a fixed rate of repetition $u = \frac{r}{d}$, adjusting $r$ to maintain $u$ as $d$ increases maintains a fixed significance requirement ($p \leq u \alpha$) and a fixed level of confidence $1 - \alpha$. If the decision points are evenly spaced over the experiment (in time or in number of observations), and the required $p$-value persists once it is achieved, then adding decision points only extends the length of the experiment before early stopping by up to $u$ of the length of the experiment. 

In our earlier example, with $u = 0.05$, using daily decision points ($d = 20$) allow us to stop the experiment when the $p$-value requirement is first met. But using hourly or even per-observation decision points only requires us to persist for up to one day (5\% of the experiment) after achieving significance for the first time, if significance continues to hold. 

And this is in the worst case: that significance begins at the hour or observation that coincides with the daily check. Otherwise, some of the 5\% of decision points with significance have already occurred by the daily check, so fewer than 5\% remain before early stopping. Similarly, if significance fails to hold at first, then holds off and on for a part of the experiment, then holds for the remainder of the experiment, then it is possible to stop earlier with more decision points. 

Note that requiring the criterion to hold with significance for a fixed fraction of the decision points in the test is different from stopping the experiment early if the criterion holds for a fixed fraction of the decision points that have occurred. The results in this section do not apply directly to that case; we address it next.

\begin{thm} \label{thm_unlimited}
Let $u$ be a fixed repetition rate that we specify, let $s$ be a minimum number of decision points for which we decide to run the test. If after the first $t$ decision points, at least a fraction $u$ of the $p$-values for a single criterion have
\be
p \leq \alpha u \frac{s}{4 t},
\ee
then the criterion holds out of sample with confidence at least $1 - \alpha$. 
\end{thm}

\begin{proof}
We will spend $\alpha$ over subtests with increasing length. Then we will apply the result from the smallest subtest of length at least $t$, applying early stopping from that subtest to our results.

Define a sequence of waypoints $w_k = 2^k s$ for $k = 1, 2, \ldots$. Waypoint $w_k$ is the endpoint for subtest $k$. Let $\alpha_k = \alpha 2^{-k}$ for $k = 1, 2, \ldots$. Let $r_k = \frac{u w_k}{2}$. Then for each subtest, if at least $r_k$ of the first $w_k$ decision points have 
\be
p \leq \frac{\alpha_k}{w_k} r_k,
\ee
then the probability that the criterion fails to hold out of sample is at most $\alpha_k$, by Theorem \ref{thm_rdm}, with $\alpha$ set to $\alpha_k$, $m = 1$ since we have a single criterion, and $d = w_k$.

Note that the sum of all $\alpha_k$ values is $\alpha$. So by the sum bound (Boole's inequality), all subtest results are valid, with probability at least $1 - \alpha$. 

Let $\hat{k}$ be the least value such that $w_{\hat{k}}$ is $t$ or greater. Then $w_{\hat{k}} < 2t$, so having significance for at least $u$ of the $t$ decision points meets the repetition requirement for the subtest, since
\be
r_{\hat{k}} = \frac{u w_{\hat{k}}}{2} < u t.
\ee

It remains to show that the required $p$-value for subtest $\hat{k}$ is at least the $p$-value in the theorem.
\be
\alpha_{\hat{k}} r_{\hat{k}}  
\ee
Expand $\alpha_{\hat{k}}$ and $r_{\hat{k}}$. 
\be
= \frac{\alpha 2^{-\hat{k}}}{w_{\hat{k}}} \frac{u w_{\hat{k}}}{2} 
\ee
Cancel $w_{\hat{k}}$ from the numerator and denominator.
\be
= \frac{1}{2} \alpha u 2^{-\hat{k}}
\ee
Multiply and divide by s.
\be
= \frac{1}{2} \alpha u \frac{s}{2^{\hat{k}} s}
\ee
Since $w_{\hat{k}} = 2^{\hat{k}} s < 2t$, this is 
\be
>  \frac{1}{2} \frac{\alpha s u}{2 t} = \alpha u \frac{s}{4 t}.
\ee
\end{proof}

With unlimited testing, stopping when $u$ of the decision points seen so far are significant for the criterion (rather than significant for $u$ of all decision points) makes the $p$-value requirements more stringent by a factor of $ \frac{s}{4 t}$. That factor has a tunable parameter $s$. Selecting it mediates a tradeoff: increasing $s$ increases the amount of testing required before early stopping is allowed, while making the $p$-value requirement less stringent. 

The following theorem is a generalization of this result. It applies to tests with pre-specified maximum sizes as well as unlimited tests. To apply it to the conjunction of multiple criteria, partition $\alpha$ over the criteria, then apply this theorem to each one separately, and combine by Boole's inequality. 

\begin{thm} \label{thm_general}
Suppose there is a single criterion and $v$ subtests, indexed by $k = 1, \ldots, v$, each with a repetition requirement $r_k$, a starting decision point index $a_k$, and an ending decision point index $b_k$. Let 
\be
\sum_{k = 1}^{v} \sum_{t = a_k}^{b_k} \alpha_{kt} = \alpha
\ee
for nonnegative values $\alpha_{kt}$ specified prior to testing. If at any decision point $t$, for any $k$, there are at least $r_k$ decision points in $a_k, \ldots, \min(t, b_k)$ with 
\be
p \leq \alpha_{kt} r_k
\ee
then the criterion holds out of sample with confidence $1 - \alpha$.
\end{thm}

\begin{proof}
For each subtest $k$, let $\alpha_k = \sum_{t = a_k}^{b_k} \alpha_{kt}$. Apply Theorem \ref{thm_aitr}, with $m = 1$ (single criterion), $\alpha = \alpha_k$, $\alpha_{1t} = \alpha_{kt}$, and $r = r_k$. Then the probability that subtest $k$ succeeds (that there are $r_k$ decision points with the specified $p$-values in decision points $[a_k, b_k]$) yet the criterion does not hold out of sample is at most $\alpha_k$. By Boole's inequality (the sum bound on the probability of a union), the probability that any subtest succeeds and the criterion fails to hold out of sample is the sum of $\alpha_k$ over the subtests:
\be
\sum_{k = 1}^{v} \alpha_k = \sum_{k = 1}^{v} \sum_{t = a_k}^{b_k} \alpha_{kt} = \alpha.
\ee
\end{proof}

Applying this theorem requires specifying subtests: repetition requirements $r_k$, start and end decision point indices $a_k$ and $b_k$, and a partition of $\alpha$ into values $\alpha_{kt}$. Theorem \ref{thm_unlimited} can be seen as a corollary of Theorem \ref{thm_general}, and it illustrates that $v$ need not be bounded. To prevent over-splitting $\alpha$, it can be useful to limit the set of subtests, for example using $b_k = 2^k$ instead of $b_k = k$ and limit the scope of subtests, for example using low values for $b_k$ for low $r_k$ values. 

Two strategies to spend $\alpha$ are uniform and geometric. Uniform spending is equal over divisions: subtests, decision points, or criteria. Geometric spending allocates a fraction of the remaining budget for each division: select a ``withdrawal rate" $0 < w < 1$ and set $\alpha_1 = w \alpha$, $\alpha_2 = w (1 - w) \alpha$, and in general $\alpha_j = w (1 - w)^{j - 1} \alpha$. This allows an unbounded number of divisions. (Theorem \ref{thm_unlimited} uses geometric spending with $w = \frac{1}{2}$.)

\section{Analysis and Type II Error Strategy} \label{sec_type2}
The methods in this paper all guarantee that the probability of type I error (false positive) is at most $\alpha$. But different choices of number of decision points and required numbers of repeats can lead to different probabilities of type II error (false negative). This section explores, in general terms, how $\alpha$ spending decisions and choices of required numbers of repeats relate to test sizes required to achieve type II error similar to that achieved without early stopping. It does so through the lens of $Z$ scores required to achieve statistical significance, because $Z$ scores are a more familiar way of thinking about significance than $p$-values for many practitioners.

\begin{figure} 
\includegraphics[width=3.5in]{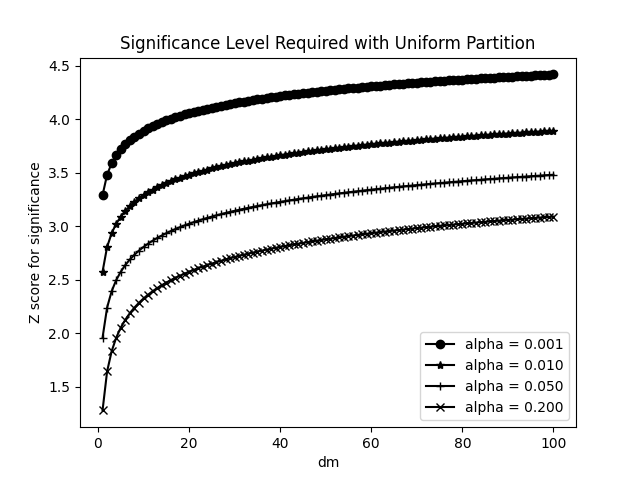} 
\caption{\textbf{$\mathbf{Z}$ score required, by $\mathbf{dm}$}. Each $Z$ score is the inverse of standard normal cdf for $1 - \frac{p}{2}$ with $p = \frac{\alpha}{dm}$. The $Z$ score required for significance under a uniform partition of $\alpha$ increases as the number of decision points $d$ and criteria $m$ increase, slicing $\alpha$ more finely. The plots begin at $dm = 1$. The $Z$ score required for $dm = 1$ is the $Z$ score without the possibility of early stopping (one decision point at the end of the test) and for a single criterion. The required $Z$ scores have decreasing marginal increases as $dm$ increases.} \label{fig_zdm}
\end{figure}

Figure \ref{fig_zdm} shows the $Z$ scores required for significance given confidence level $1 - \alpha$, $d$ decision points, and $m$ criteria, with $\alpha$ partitioned uniformly over the $d m$ criteria by decision points, under the approximation that the criteria metrics have normal distributions. So $Z$ is the inverse of the standard normal cdf for $1 - \frac{p}{2}$, with $p$ the required $p$-value for significance: $p = \frac{\alpha}{d m}$. (The 2 in $1 - \frac{p}{2}$ reflects two-sided criteria.) 

Note that for $\alpha = 0.05$ and $d m = 1$, the required $Z$ score is approximately 2. This corresponds to the well-known 2 SEM (standard error of the mean) rule for tests with 95\% confidence ($\alpha = 5\%$) and a single criterion ($m = 1$) and one decision point ($d = 1$), at the end of the test. (In the plotted lines, the first value for $d m$ is one.) As $d m$ increases, $\alpha$ is partitioned more finely, reducing the $p$-values, and hence increasing the $Z$ values, required for significance. 

Consider a criterion that the mean, over the observations collected in the test, of some metric is significantly greater than or less than zero. Let $\mu$ be the mean of the distribution that generates the observations, and let $\sigma$ be the standard deviation. Let $n$ be the number of observations collected prior to a decision point, and let $\hmu$ be the mean and $\hsigma$ be the standard deviation over those observations. If 
\be
\left| \hmu \right| = \frac{Z \hsigma}{\sqrt{n}} 
\ee
then the $Z$ score at the decision point is $Z$. Solve for $n$:
\be
n = Z^2 \left(\frac{\hsigma}{\hmu}\right)^2.
\ee
For large numbers of observations, $\hmu \approx \mu$ and $\hsigma \approx \sigma$, and
\be
n \approx Z^2 \left(\frac{\sigma}{\mu}\right)^2.
\ee
So the test size (in observations) required to achieve significance tends to be loosely proportional to $Z^2$. (Loosely, because we are not accounting for $\hmu$ and $\hsigma$ not being exactly $\mu$ and $\sigma$, respectively, or for $\hmu$ and $\hsigma$ tending to be less noisy, and approximate $\mu$ and $\sigma$ more closely, for more observations.) As a result, if we double the $Z$ score required for significance, for example, then the probability of type II error becomes that of a test that has only about a quarter as many observations. (See Figure \ref{fig_zp}.)

\begin{figure} 
\includegraphics[width=3.5in]{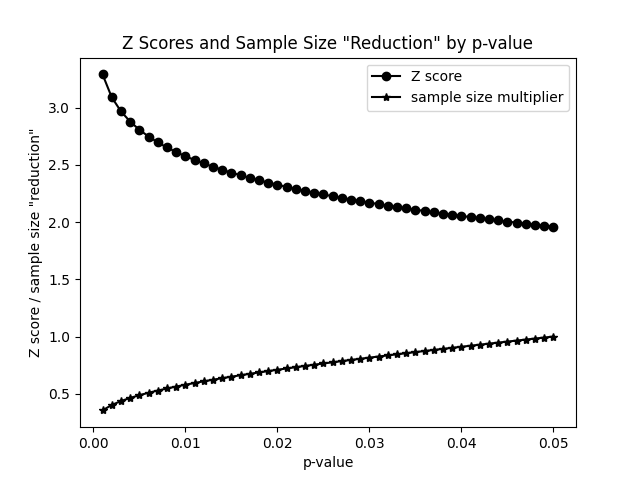} 
\caption{\textbf{$\mathbf{Z}$ score and sample size impact for partitioning $\mathbf{\alpha}$}. The top curve is the $Z$ score required for statistical significance given the $p$-value required for statistical significance on the $x$-axis. The bottom curve is the approximate reduction in sample size (compared to $p = 0.05$) for the purpose of achieving statistical significance to avoid type II error. (Each point in the lower curve is the square of the ratio of the $Z$ score for $p = 0.05$ to the $Z$-score for the $p$-value on the $x$-axis.} \label{fig_zp}
\end{figure}

Suppose a test has a single criterion ($m = 1$) and we want 95\% confidence $\alpha = 0.05$. For a single decision point at the end of the test (no early stopping), we have $Z \approx 1.96$ required for significance. If, instead, we use 20 decision points with uniform $\alpha$ spending, then $Z \approx 3.02$ (See Figure \ref{fig_zdm}.) In exchange for allowing early stopping at 19 points in addition to the end of the test, the probability of type II error becomes approximately that of a test $(1.96/3.02)^2 \approx 40\%$ as long without early stopping. 

One strategy to reduce the probability of type II error while maintaining the possibility of early stopping with confidence is to spend a large portion of $\alpha$ on the decision point at the end of the test, and spend $\alpha$ uniformly over the other decision points. (This is the general strategy of the Haybittle-Peto boundary for sequential testing \cite{haybittle71,peto76}.) For example, suppose we allocate half of $\alpha = 0.05$ to the final decision point and the other half equally among 19 other decision points. For the final point, $Z \approx 2.24$ (the second value for the $\alpha = 0.05$ line in Figure \ref{fig_zdm}), giving approximately the probability of type II error of a test $(1.96/2.24)^2 \approx 77\%$ as long without early stopping. For the 19 other decision points, the required $Z$ score for early stopping becomes approximately $3.21$ rather than the $3.02$ for uniform spending with $d = 20$. Since $(3.21/3.02)^2 \approx 1.13$, it may take approximately 13\% longer to achieve this higher $Z$ score in order to stop early. 

To achieve similar results for a very large number of decision points, allocate a share of $\alpha$ (say $\theta \alpha$ for $0 < \theta < 1$) to the final decision point, then allocate the remainder to the other points, and require repetition over the other decision points but not over the final decision point. Then the test is a success if there is a significant result if for each criterion either ($p \leq \frac{\theta \alpha}{m}$) at the final decision point or at least $r$ of the other decision points have $p \leq \frac{(1 - \theta) \alpha r}{(d - 1) m}$. Increasing $\theta$ decreases the probability of type II error but de-emphasizes early stopping. 

Figure \ref{fig_zu} shows $Z$ values required for significance if a portion $u = \frac{r}{d}$ of the decision points are required to have significant results for test success, based on uniform $\alpha$ spending over decision points, for a single criterion. Requiring $u$ of the results at decision points to be significant results gives the same $Z$ score requirements as having $d = \frac{1}{u}$ decision points. So, for example, on the $\alpha = 0.05$ line in Figure \ref{fig_zu}, for $u = 0.05$, $Z \approx 3.02$ -- the same value as for $d m = 20$ on the $\alpha = 0.05$ line in Figure \ref{fig_zdm}. 

\begin{figure} 
\includegraphics[width=3.5in]{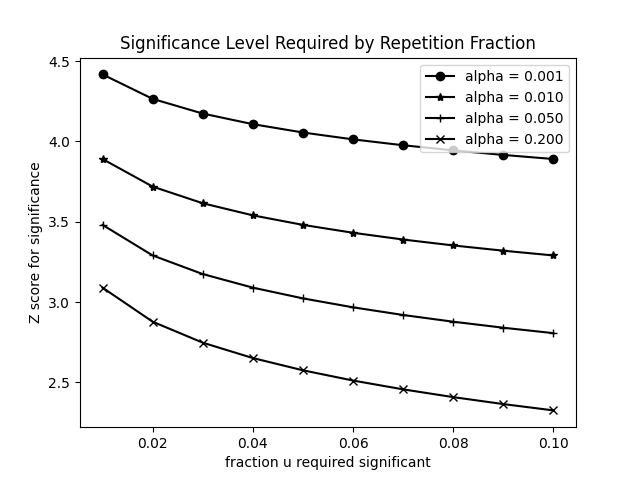} 
\caption{\textbf{$\mathbf{Z}$ score required, by $\mathbf{u}$}. The $Z$ score required for significance varies with the fraction of decision points $u$ at which we require significance in order to stop the test with $1 - \alpha$ confidence that a criterion holds out of sample. For comparison, without early stopping, and with $\alpha = 0.05$, $Z \approx 2$ is the required score. Requiring significance at 5\% of decision points ($u = 0.05$) makes the required score $Z \approx 3$. Requiring a larger fraction of significant results has (gently) decreasing marginal returns in reducing required $Z$ scores.} \label{fig_zu}
\end{figure}

Now we will briefly compare the method of requiring repeated significance to a recent always-valid method \cite{waudbysmith24} (their Equation 8, page 8). That method requires $Z$ score 
\be
\sqrt{\frac{2(t \rho^2 + 1)}{t \rho^2} \ln \left(\frac{\sqrt{t \rho^2 + 1}}{\alpha}\right)}
\ee
and has a parameter $\rho$ that controls the value of $t$ at which the required $Z$ score is minimized. This method has the advantage of not requiring repeated significance to stop early, and it does not require the test plan to specify a test size -- the test may continue indefinitely. Figure \ref{fig_zrhovu} compares it to requiring a fraction $u$ of the decision points to have 
\be
p \leq \alpha u,
\ee
with the number of observations and decision points determined before testing. Requiring 5\% of decision points to have significance matches the performance of the always-valid method. If there are 20 decision points, then that allows early stopping at the first significant result. For a decision point at every observation, 100,000 of the 20 million observations would need significant results to allow early stopping with confidence. 

\begin{figure} 
\includegraphics[width=3.5in]{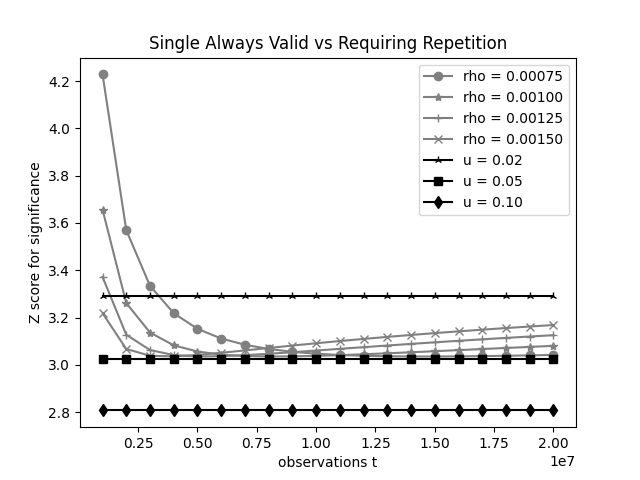} 
\caption{\textbf{$\mathbf{Z}$ score required, continuous monitoring vs. requiring repetition}. Comparison of significance requirements for early stopping for (1) a continuous monitoring method \cite{waudbysmith24} that requires only one significant result and relies on each observation having limited impact on the $p$-value, and (2) requiring repeated significance. For this example: 20 million observations, up to 20 million decision points, a single criterion, and $\alpha = 0.05$, requiring significance for 5\% of the decision points matches the performance of the method that requires only a single significant result.} \label{fig_zrhovu}
\end{figure} 

\section{Discussion} \label{sec_disc}
This paper introduces the tool of requiring repeated significance to allow less-stringent requirements for significance during testing. We showed that doing so enables effective tests with many decision points, even decision points on a per-observation basis, to achieve continuous monitoring.

Several ideas in this paper have their roots in machine learning. Fixing a portion of decision points to set the number of required repetitions is similar to work in machine learning showing that selecting a fixed portion of classifiers for an equally-weighted Gibbs ensemble classifier has similar error bounds to selecting from a set of hypothesis classifiers with size the inverse of the portion, even if the actual hypothesis class is infinite \cite{bax16}. The method in Section \ref{sec_cm} for dealing with a test length that is unknown a priori is conceptually similar to \textit{luckiness frameworks} \cite{shawe-taylor98,shawe-taylor04}, which budget $\alpha$ in ever-thinner amounts over ever-larger hypothesis classes, and thus can supply error bounds for training over unbounded hypothesis classes.

In Section \ref{sec_type2}, we rely on a normal approximation to convert between $p$-values and $Z$ scores. The purpose of using a normal approximation was to analyze how different numbers of observations relate to achieving different levels of significance in a general way. (The results in the other sections do not rely on normal approximations.) That said, bucket testing systems that use large numbers of observations, for example to tune website features, often use normal approximations to produce their $p$-values by first computing $Z$ scores based on empirical means and standard deviations. 

Other methods to compute $p$-values include binomial tail inversion \cite{hoel54,langford05} and hypergeometric tail inversion \cite{chvatal79,skala13,bax_donald19} for rate-based criteria and using standard statistical software to compute cdf values for metrics with known distributions. In many testing systems, $p$-values are approximate, often using a normal approximation that treats empirical variance as the variance of the distribution. For non-approximate testing, consider using concentration inequalities: Bennet \cite{bennett62}, Hoeffding \cite{hoeffding63} or McDiarmid \cite{mcdiarmid89} bounds, or empirical Bernstein bounds \cite{maurer09}, solving for the probability of bound failure with bound range the boundary for a criterion to produce the equivalent of a $p$-value. It is worthwhile to understand how your testing system computes $p$-values, to evaluate how to use them effectively. 

Similarly, many systems assume that test observations are drawn i.i.d. from a (possibly unknown) underlying distribution. It is advisable to examine that assumption and adjust testing accordingly. For example, for many websites users behave differently on weekends than on weekdays, so a suitable test plan may wait for a week, or at least one weekend day and one weekday, before the first decision point. And it may be worthwhile to err on the side of conservatism with confidence, knowing that the i.i.d. assumption underlying $p$-value computations does not quite hold, because of time-to-effect \cite{pocock05}, time-series, and periodic effects . Also, it is useful to realize that even with bucket testing, there is a selection effect that, on average, makes out-of-sample results inferior to test results \cite{bax_compare_prices,pocock05,wang16,huang17}. 

Some continuous monitoring methods \cite{waudbysmith24} rely on individual observations making small changes to metrics and those changes tending to shrink as a test progresses. For example, for a metric that averages over 0-1 (Bernoulli) variables, the running average for observation $t$ is $\frac{t - 1}{t}$ times the previous running average, plus $\frac{1}{t}$ if the observation value is one, so later observations affect the running average less than earlier observations. Such properties imply that metric values become highly correlated from one observation to the next as the test progresses. 

The methods in this paper do not rely on such an assumption; instead, by requiring repetition, they test to what extent such an assumption holds in practice. That enables us to apply the methods in this paper with fewer assumptions about criteria metrics, and without needing to use part of the sample or the $\alpha$ budget to verify those assumptions. In the future, it would be interesting to try to blend the two approaches in hopes of allowing earlier stopping with confidence for the most common types of metrics -- averages over i.i.d. observations that occur over the course of the test.

To be clear, though, many metrics do not fit the model of a simple average over i.i.d. observations that continue to accrue during the test. For example, an average over website users of per-user averages does not fit that model precisely. Additional observations for regular users make smaller contributions to such a metric than initial observations for occasional users late in a test. And the total number of users that is the denominator for the metric may increase quickly at the start of the test then slowly or not at all later in the test. The approach in this paper, by not assuming an underlying model for how observations contribute to metrics over time, allows application without needing to evaluate whether such a model holds, though it does require that the $p$-values apply to the metric of interest. 

In the future, it would also be interesting to investigate ways to defer decisions about how to spend some of $\alpha$ until a portion of the test is completed and still maintain test integrity. We learn from testing, so it would be useful to apply what we learn within the test, to the extent we can do so and still maintain type I error guarantees.

%***Future Work*** Dealing with stopping when a fixed portion of the decision points so far have significant results for each criterion. Optimizing bounds. Just-in-time error spending -- must not use data from before the new plan, must decide a priori on re-plan point (time or sample size), can decide a priori to fork the experiment and how much $\alpha$ to give to each sub-experiment, can have a sub-experiment that is a continuance of the old experiment, but must allocate its portion of $\alpha$ a priori. Given all that, can select among the sub-experiments for results. For multiple subtests, if a criterion meets a $p$-value requirement at a decision point, then it is also significant for higher $p$-values at the same decision point. Does this impose constraints on the distribution of type-I errors across subtests for the same criterion and decision point? Can those constraints be used to derive early-stopping rules with less stringent $p$-value requirements? 

\bibliographystyle{unsrt}
\bibliography{bax}

\end{document}